\documentclass[aps,twocolumn,superscriptaddress,longbibliography]{revtex4}
\usepackage[utf8]{inputenc}
\usepackage{bm}
\usepackage{times}
 \usepackage{algorithm,algorithmic}

\usepackage{mathtools}
\usepackage{amsmath}
\usepackage[shortlabels]{enumitem}
\usepackage{graphicx,epic,eepic,epsfig,amsmath,latexsym,amssymb,verbatim,color}
\usepackage{amsfonts}       
\usepackage{nicefrac}       
\usepackage{amsmath}
\usepackage{bbm}
\usepackage{float}
\usepackage{tikz}
\usetikzlibrary{chains}
\usetikzlibrary{fit}
\usepackage{pgflibraryarrows}		
\usepackage{pgflibrarysnakes}		

\usepackage{epsfig}
\usetikzlibrary{shapes.symbols,patterns} 
\usepackage{pgfplots}

\usepackage[strict]{changepage}
\usepackage{hyperref}
\hypersetup{colorlinks=true,citecolor=blue,linkcolor=blue,filecolor=blue,urlcolor=blue,breaklinks=true}

\usepackage[marginal]{footmisc}
\usepackage{url}
\usepackage{theorem}

\newtheorem{proposition}{Proposition}
\newtheorem{lemma}[proposition]{Lemma}

\newtheorem{theorem}[proposition]{Theorem}

\newtheorem{corollary}[proposition]{Corollary}

\def\squareforqed{\hbox{\rlap{$\sqcap$}$\sqcup$}}
\def\qed{\ifmmode\squareforqed\else{\unskip\nobreak\hfil
\penalty50\hskip1em\null\nobreak\hfil\squareforqed
\parfillskip=0pt\finalhyphendemerits=0\endgraf}\fi}
\def\endenv{\ifmmode\;\else{\unskip\nobreak\hfil
\penalty50\hskip1em\null\nobreak\hfil\;
\parfillskip=0pt\finalhyphendemerits=0\endgraf}\fi}
\newenvironment{proof}{\noindent \textbf{{Proof~} }}{\hfill $\blacksquare$}
\newcounter{remark}

\newcounter{example}

\mathchardef\ordinarycolon\mathcode`\:
\mathcode`\:=\string"8000
\def\vcentcolon{\mathrel{\mathop\ordinarycolon}}
\begingroup \catcode`\:=\active
  \lowercase{\endgroup
  \let :\vcentcolon
  }

\usepackage{cleveref}
\usepackage{graphicx}
\usepackage{xcolor}

\RequirePackage[framemethod=default]{mdframed}
\newmdenv[skipabove=7pt,
skipbelow=7pt,
backgroundcolor=darkblue!15,
innerleftmargin=5pt,
innerrightmargin=5pt,
innertopmargin=5pt,
leftmargin=0cm,
rightmargin=0cm,
innerbottommargin=5pt,
linewidth=1pt]{tBox}

\definecolor{darkblue}{RGB}{0,76,156}
\definecolor{darkkblue}{RGB}{0,0,153}
\definecolor{blue2}{RGB}{102,178,255}
\definecolor{darkred}{RGB}{195,0,0}
\newcommand{\nc}{\newcommand}
\nc{\rnc}{\renewcommand}
\nc{\beg}{\begin{equation}}
\nc{\eeq}{{\end{equation}}}
\nc{\beqa}{\begin{eqnarray}}
\nc{\eeqa}{\end{eqnarray}}
\nc{\lbar}[1]{\overline{#1}}
\nc{\bra}[1]{\langle#1|}
\nc{\ket}[1]{|#1\rangle}
\nc{\ketbra}[2]{|#1\rangle\!\langle#2|}
\nc{\braket}[2]{\langle#1|#2\rangle}

\nc{\proj}[1]{| #1\rangle\!\langle #1 |}
\nc{\avg}[1]{\langle#1\rangle}
\nc{\rank}{\operatorname{Rank}}
\nc{\smfrac}[2]{\mbox{$\frac{#1}{#2}$}}
\nc{\tr}{\operatorname{Tr}}
\nc{\ox}{\otimes}
\nc{\dg}{\dagger}
\nc{\dn}{\downarrow}
\nc{\cA}{{\cal A}}
\nc{\cB}{{\cal B}}
\nc{\cC}{{\cal C}}
\nc{\cD}{{\cal D}}
\nc{\cE}{{\cal E}}
\nc{\cF}{{\cal F}}
\nc{\cG}{{\cal G}}
\nc{\cH}{{\cal H}}
\nc{\cI}{{\cal I}}
\nc{\cJ}{{\cal J}}
\nc{\cK}{{\cal K}}
\nc{\cL}{{\cal L}}
\nc{\cM}{{\cal M}}
\nc{\cN}{{\cal N}}
\nc{\cO}{{\cal O}}
\nc{\cP}{{\cal P}}
\nc{\cQ}{{\cal Q}}
\nc{\cR}{{\cal R}}
\nc{\cS}{{\cal S}}
\nc{\cT}{{\cal T}}
\nc{\cV}{{\cal V}}
\nc{\cX}{{\cal X}}
\nc{\cY}{{\cal Y}}
\nc{\cZ}{{\cal Z}}
\nc{\cW}{{\cal W}}
\nc{\csupp}{{\operatorname{csupp}}}
\nc{\qsupp}{{\operatorname{qsupp}}}
\nc{\var}{{\operatorname{var}}}
\nc{\rar}{\rightarrow}
\nc{\lrar}{\longrightarrow}
\nc{\polylog}{{\operatorname{polylog}}}
\nc{\wt}{{\operatorname{wt}}}
\nc{\av}[1]{{\left\langle {#1} \right\rangle}}
\nc{\supp}{{\operatorname{supp}}}

\nc{\argmin}{{\operatorname{argmin}}}

\def\t{\theta}

\def\x{\xi}

\nc{\RR}{{{\mathbb R}}}
\nc{\CC}{{{\mathbb C}}}
\nc{\FF}{{{\mathbb F}}}
\nc{\NN}{{{\mathbb N}}}
\nc{\ZZ}{{{\mathbb Z}}}
\nc{\PP}{{{\mathbb P}}}
\nc{\QQ}{{{\mathbb Q}}}
\nc{\UU}{{{\mathbb U}}}
\nc{\EE}{{{\mathbb E}}}
\nc{\id}{{\operatorname{id}}}

\nc{\CHSH}{{\operatorname{CHSH}}}

\nc{\be}{\begin{equation}}
\nc{\ee}{{\end{equation}}}
\nc{\bea}{\begin{eqnarray}}
\nc{\eea}{\end{eqnarray}}
\nc{\<}{\langle}
\rnc{\>}{\rangle}
\nc{\rU}{\mbox{U}}

\nc{\ob}[1]{#1}

\nc{\SEP}{{\text{\rm SEP}}}
\nc{\NS}{{\text{\rm NS}}}
\nc{\LOCC}{{\text{\rm LOCC}}}
\nc{\PPT}{{\text{\rm PPT}}}
\nc{\EXT}{{\text{\rm EXT}}}
\nc{\Sym}{{\operatorname{Sym}}}

\nc{\ERLO}{{E_{\text{r,LO}}}}
\nc{\ERLOCC}{{E_{\text{r,LOCC}}}}
\nc{\ERPPT}{{E_{\text{r,PPT}}}}
\nc{\ERLOCCinfty}{{E^{\infty}_{\text{r,LOCC}}}}
\nc{\Aram}{{\operatorname{\sf A}}}

\usepackage{tikz}
\usepackage{hyperref}
\hypersetup{colorlinks=true,citecolor=blue,linkcolor=blue,filecolor=blue,urlcolor=blue,breaklinks=true}

\makeatletter
\def\grd@save@target#1{%
  \def\grd@target{#1}}
\def\grd@save@start#1{%
  \def\grd@start{#1}}
\tikzset{
  grid with coordinates/.style={
    to path={%
      \pgfextra{%
        \edef\grd@@target{(\tikztotarget)}%
        \tikz@scan@one@point\grd@save@target\grd@@target\relax
        \edef\grd@@start{(\tikztostart)}%
        \tikz@scan@one@point\grd@save@start\grd@@start\relax
        \draw[minor help lines,magenta] (\tikztostart) grid (\tikztotarget);
        \draw[major help lines] (\tikztostart) grid (\tikztotarget);
        \grd@start
        \pgfmathsetmacro{\grd@xa}{\the\pgf@x/1cm}
        \pgfmathsetmacro{\grd@ya}{\the\pgf@y/1cm}
        \grd@target
        \pgfmathsetmacro{\grd@xb}{\the\pgf@x/1cm}
        \pgfmathsetmacro{\grd@yb}{\the\pgf@y/1cm}
        \pgfmathsetmacro{\grd@xc}{\grd@xa + \pgfkeysvalueof{/tikz/grid with coordinates/major step}}
        \pgfmathsetmacro{\grd@yc}{\grd@ya + \pgfkeysvalueof{/tikz/grid with coordinates/major step}}
        \foreach \x in {\grd@xa,\grd@xc,...,\grd@xb}
        \node[anchor=north] at (\x,\grd@ya) {\pgfmathprintnumber{\x}};
        \foreach \y in {\grd@ya,\grd@yc,...,\grd@yb}
        \node[anchor=east] at (\grd@xa,\y) {\pgfmathprintnumber{\y}};
      }
    }
  },
  minor help lines/.style={
    help lines,
    step=\pgfkeysvalueof{/tikz/grid with coordinates/minor step}
  },
  major help lines/.style={
    help lines,
    line width=\pgfkeysvalueof{/tikz/grid with coordinates/major line width},
    step=\pgfkeysvalueof{/tikz/grid with coordinates/major step}
  },
  grid with coordinates/.cd,
  minor step/.initial=.2,
  major step/.initial=1,
  major line width/.initial=2pt,
}
\makeatother

\usepackage{thmtools}
\usepackage{thm-restate}
\usepackage{etoolbox}
\makeatletter
\def\problem@s{}
\newcounter{problems@cnt}

\newcommand{\allproblems}{\problem@s}
\makeatother
\definecolor{colorone}{rgb}{1,0.36,0.03}
\definecolor{colortwo}{rgb}{0.54,0.71,0.03}
\definecolor{colorthree}{rgb}{0.01,0.51,0.93}
\definecolor{colorfour}{rgb}{0.47,0.26,0.58}
\definecolor{red2}{rgb}{0.8,0.1,0.1}
\usepackage{tcolorbox}
\usepackage{relsize}
\usepackage{graphicx}
\usepackage{subfigure}
\nc{\st}{\text{subject to} \ }
\nc{\supre}{\text{supremum} \ }
\nc{\sdp}{\text{sdp}}
\nc{\cU}{\mathcal U}
\usepackage[qm]{qcircuit}
\usepackage{array}

\newcommand{\update}[1]{\textcolor{black}{#1}}

\begin{document}
\title{Variational Quantum Algorithms for Trace Distance and Fidelity Estimation}
\author{Ranyiliu Chen}
\author{Zhixin Song}
\author{Xuanqiang Zhao}
\author{Xin Wang}
\email{wangxin73@baidu.com}
\affiliation{Institute for Quantum Computing, Baidu Research, Beijing 100193, China}
\begin{abstract} 
Estimating the difference between quantum data is crucial in quantum computing. However, as typical characterizations of quantum data similarity, the trace distance and quantum fidelity are believed to be exponentially-hard to evaluate in general. In this work, we introduce hybrid quantum-classical algorithms for these two distance measures on near-term quantum devices where no assumption of input state is required. First, we introduce the Variational Trace Distance Estimation (VTDE) algorithm. We in particular provide the technique to extract the desired spectrum information of any Hermitian matrix by local measurement. A novel variational algorithm for trace distance estimation is then derived from this technique, with the assistance of a single ancillary qubit. Notably, VTDE could avoid the barren plateau issue with logarithmic depth circuits due to a local cost function. Second, we introduce the Variational Fidelity Estimation (VFE) algorithm. We combine Uhlmann's theorem and the freedom in purification to translate the estimation task into an optimization problem over a unitary on an ancillary system with fixed purified inputs. We then provide a purification subroutine to complete the translation. Both algorithms are verified by numerical simulations and experimental implementations, exhibiting high accuracy for randomly generated mixed states.
\end{abstract}  

\date{\today}
\maketitle

\section{Introduction}
With surging advances in material, manufacturing, and quantum control, quantum computing has been driven into the noisy intermediate-scale quantum (NISQ) era~\cite{Preskill2018}, which requires novel algorithms running on a limited number of qubits with unwanted interference of the environment. The hybrid quantum-classical computation framework~\cite{McClean2016} is regarded as well-suited for execution on NISQ devices and is expected to show practical near-term applications in quantum chemistry and quantum machine learning~\cite{Biamonte2017a,Schuld2015a,Arunachalam2017}. Specifically, hybrid quantum-classical algorithms utilize the parameterized quantum circuits (PQCs)~\cite{Benedetti2019a} and classical optimization to solve problems. Such idea was applied to many key areas including Hamiltonian ground and excited states preparation~\cite{Peruzzo2014,Nakanishi2019}, quantum compiling~\cite{Sharma2020}, quantum classification~\cite{Schuld2018,Grant2018,Li2021}, Gibbs state preparation~\cite{Yuan2018a,Wu2019b,Wang2020a,Chowdhury2020}, entanglement manipulation~\cite{Zhao2021}, and quantum linear algebra~\cite{Xu2019,Bravo-Prieto2019,Huang2019b,Wang2020d,Bravo-Prieto2019a}. We refer to~\cite{Endo2020, cerezo2021variational, bharti2021noisy} for a detailed review.

Applications mentioned above and many other quantum information tasks suffer from unwanted interactions with the environment when implemented by real-world quantum systems, leading to errors in their working qubits. Thus, metric estimation for quantum states is vital to benchmark the tasks’ implementation. As scalable quantum computers and quantum error correction are still on their way~\cite{Terhal2015}, estimating metric on a near-term device is essential to the verification of quantum information processing tasks and quantify how well quantum information has been preserved. Moreover, metric estimation is also an integral part of quantum machine learning~\cite{Biamonte2017a,Schuld2015a,Arunachalam2017}. For example, metrics could play the role of the loss function in state learning tasks \cite{Hu_2019,braccia2020enhance}.

The trace distance and fidelity are two typical metrics to quantify how close two quantum states are \cite{Nielsen2010,Jozsa1994,Uhlmann1976}.
Given two quantum states $\rho$ and $\sigma$, the trace distance $D(\rho,\sigma)$ and the fidelity $F(\rho,\sigma)$ are defined as follows:
    \begin{align}
D(\rho, \sigma)&:=\frac{1}{2}\|\rho-\sigma\|_1,\\
F(\rho, \sigma)&:= \tr \sqrt{\sqrt{\rho} \sigma \sqrt{\rho}}
= \Vert \sqrt{\rho} \sqrt{\sigma} \Vert_{1},
    \end{align}
where $\|\cdot\|_1$ denotes the trace norm. 
When at least one of the states is pure, the task of fidelity estimation reduces to the simple case of calculating the square root of the state overlap $F(\rho,\sigma) 
=\sqrt{\tr \rho\sigma}$ which can be obtained by the Swap-test \cite{Buhrman_2001}. Thus the trace distance in this case is bounded by $1-F(\rho,\sigma)\le D(\rho,\sigma)\le\sqrt{1-F(\rho,\sigma)^2}$. However, evaluating these two metrics for mixed states is hard in general. One might attempt to classically compute the two metric from the matrix description of quantum states $\rho$ and $\sigma$ obtained via quantum state tomography. Nevertheless, this approach is infeasible due to the exponential growth of the matrix dimension with the number of qubits. On quantum computers, evaluating the trace distances is probably hard since even judging whether $\rho$ and $\sigma$ have large or small trace distance is known to be QSZK-complete~\cite{watrous2008quantum}, where QSZK
(quantum statistical zero-knowledge) is a complexity
class that includes BQP (bounded-error quantum
polynomial time). Hence estimating fidelity and trace distance could probably be hard even for quantum computers.

Several approaches have been proposed for the trace distance and fidelity estimation, and here we focus on the case where the metric is between two unknown general quantum states. In \cite{Cerezo2019}, the authors variationally estimate the truncated fidelity via a hybrid classical-quantum algorithm. The truncated fidelity bounds the exact fidelity and is a good approximation of it when one of the states is known to have a low rank. For the estimation of trace distance, several methods are proposed but only applicable in specific quantum environments \cite{Zhang_2019,Smirne_2011}. To the best of our knowledge, no methods for estimating the trace distance on general NISQ devices has been proposed yet.

To overcome these challenges, we propose the Variational Trace Distance Estimation (VTDE) algorithm as well as the Variational Fidelity Estimation (VFE) algorithm. 
First, we propose a method to estimate the trace norm of an arbitrary Hermitian matrix $H$, and apply it to trace distance estimation by specifying $H=\frac{1}{2}(\rho-\sigma)$. In particular, we prove that local measurements on an ancillary single qubit can extract the desired spectrum information of any Hermitian $H$ conjugated by a unitary. By optimizing over all unitaries, we can obtain the trace norm of $H$. Our method notably only employs a local observable in the loss function evaluation, which saves us from the gradient vanishing issue with shallow circuits.
Second, we introduce a method to estimate the fidelity of general quantum states, which utilizes the Uhlmann's theorem. We observe that the fidelity can be estimated by optimizing over all unitaries on an ancillary system via the freedom in purification. We also introduce a subroutine that works on NISQ devices to purify quantum states. With the performance analysis of the purification subroutine, we show that only few ancillary qubits are required if the unknown states are low-rank. 

This paper is organized as follows. In Sec.~\ref{sec:trace distance}, we introduce the variational quantum algorithms for trace norm and trace distance estimation.  
In Sec.~\ref{sec:fidelity}, we introduce the variational quantum algorithms for quantum state fidelity estimation and its purification subroutine for mixed state learning. Numerical experiments and experimental implementations on IMB superconducting device are provided to show the validity of our methods. We finally deliver concluding remarks in Sec.~\ref{sec:conclusion}.

\section{Variational Trace Distance Estimation}\label{sec:trace distance}
This section introduces a variational quantum algorithm for estimating the trace norm of an arbitrary Hermitian matrix $H$, which could be easily applied to trace distance estimation. Our method employs the optimization of a parameterized quantum circuit (PQC), requiring only one ancillary qubit initialized in an arbitrary pure state (typically $\proj{0}$ in practice) and single-qubit measurements. In this sense, our algorithm is practical and efficient for NISQ devices.

We will frequently use symbols such as $\mathcal{H}_A$ and $\mathcal{H}_B$ to denote Hilbert spaces associated with quantum systems $A$ and $B$, respectively. We use $d_A$ to denote the dimension of system $A$. The set of linear operators acting on $A$ is denoted by $\cL(\mathcal{H}_A)$. We usually write an operator with a subscript indicating the system that the operator acts on, such as $M_{AB}$, and write $M_A:=\tr_B M_{AB}$. Note that for a linear operator $X\in\cL(\mathcal{H}_A)$, we define $|X|=\sqrt{X^\dagger X}$, and the trace norm of $X$ is given by $\|X\|_1=\tr |X|$.

We first formulate the theory for the trace norm estimation, then we describe in detail the process of our algorithm, followed by the numerical experiments.

\subsection{Estimating trace norm via one-qubit overlap maximization}
Suppose on system $\mathcal{H}_A\otimes\mathcal{H}_B$ we have a Hermitian matrix $H_{AB}$ that has spectral decomposition 
\begin{equation}
    H_{AB} = \sum_{j=1}^{d_Ad_B}h^j_{AB}\proj{\psi_j}
    \label{eq:spedec}
\end{equation}
with decreasing spectrum $\{h^j_{AB}\}_{j=1}^{d_Ad_B}$ and orthonormal basis $\{\ket{\psi_j}\}$. Here, for technical convenience, we simply pad the spectrum with 0s to ensure the expression in Eq.~\eqref{eq:spedec}. 

Before showing the main algorithm, we introduce an optimization method to obtain some information of the spectrum of a given Hermitian matrix as follows.
\begin{proposition}\label{prop:H sum of some eig}
For any Hermitian matrix $H_{AB}\in\cL{(\mathcal{H}_A\otimes\mathcal{H}_B)}$ with spectral decomposition as Eq.~\eqref{eq:spedec}, respectively denote the dimension of $\mathcal{H}_A$, $\mathcal{H}_B$ by $d_A$, $d_B$. It holds that
\begin{equation}
\max_{U} \tr \proj{0}_A \widetilde{H}_A =\sum_{j=1}^{d_B}h^j_{AB},
\end{equation}
where $\widetilde{H}_A=\tr_B \widetilde{H}_{AB}$, $\widetilde{H}_{AB} = UH_{AB}U^\dagger$, and the optimization is over all unitaries.
\label{prop:AB}
 
\end{proposition}
This Proposition indicates that the optimization over unitaries can estimate the sum of several largest eigenvalues of $H_{AB}$. The proof of Proposition \ref{prop:AB} can be found in Appendix \ref{proof:H sum of some eig}.

To apply Proposition \ref{prop:AB} to estimate the trace norm, we fix subsystem $A$ to be single-qubit. In this case, for arbitrary Hermitian $H_{AB}$ on $d=2^n$-dimensional quantum system $\mathcal{H}_A\otimes\mathcal{H}_B$, where $\mathcal{H}_A=\CC^2$ (hence $\mathcal{H}_B=\CC^{2^{n-1}}$), the maximal expectation of measurements on system $A$ gives us the sum of the first half eigenvalues of $H_{AB}$. It is now quite close to the trace norm of $H_{AB}$, which is the sum of absolute values of its eigenvalues. We take the final step to estimate the trace norm by appending a 1-qubit pure state and optimizing twice. Concretely, we derive the following theorem to ensure the validity of our algorithm.
\begin{theorem}\label{thm:tr}
For any Hermitian $H_A$ on $n$-qubit system $A$, and any single-qubit pure state $\ket{r}$ on system $R$, it holds that
\begin{align}
\|H_A\|_1 = \max_{U^+} \tr \proj{0}_R Q_R^++\max_{U^-}\tr \proj{0}_R Q_R^-,
\end{align}
where $Q^\pm_R = \tr_AQ^\pm_{AR}$, $Q_{AR}^\pm = U^\pm (\pm H_A \otimes \proj{r}_R) U^{\pm\dagger}$, and each optimization is over unitaries on system $AR$.
\end{theorem}
This theorem shows how to generally evaluate the trace norm of arbitrary $H$ by optimization. The proof of Theorem \ref{thm:tr} can be found in Appendix \ref{proof:tr}.

We note that, since the target Hermitian matrix is in many cases (and can always be) written as the linear combination of density matrix with real coefficients $c^j$:
\begin{equation}
    H_{A}=\sum_jc^j\rho_{A}^j,
    \label{eq:hrho}
\end{equation}
we have $\tr H_A=\sum_jc^j$. Employing the information of $\tr H_A$ we can save one optimization as the following Corollary shows:
\begin{corollary}\label{corollary:trace norm}
For any Hermitian $H_A$ (written as Eq. \eqref{eq:hrho}) on $n$-qubit system $A$, and any single-qubit pure state $\ket{r}$ on system $R$, it holds that
\begin{equation}
    \begin{aligned}
        \|H_A\|_1=&\sum_{h_A^j>0}h_A^j+\sum_{h_A^j<0}-h_A^j\\
        =&2\sum_{h_A^j>0}h_A^j-\tr H_A\\
        =&2\max_{U} \tr \proj{0}_R Q_R-\sum_jc^j,
    \end{aligned}
    \label{eq:cor3}
\end{equation}
where $Q_R=\tr_A Q_{AR}$, $Q_{AR} = U(H_A\otimes \proj{r}_R)U^{\dagger}$, and the optimization is over all unitaries on system $AR$.
\label{cor:oneopt}
\end{corollary}

Also in this case, $\max_{U}\tr \proj{0}_R Q_R$ in practice can be evaluated by post-processing: 
\begin{equation}
    \tr \proj{0}_R Q_R=\sum_jc^j\tr \proj{0}_R(U\rho^j_{A}\otimes\proj{r}_RU^\dagger)_R.
\end{equation}
This means that we could use PQC, taking $\rho^j_A$ as inputs, to perform the optimization for the trace norm estimation. In next section we will adopt this strategy to estimate the trace distance, for which $H=\frac{1}{2}(\rho-\sigma)$.

\subsection{Variational trace distance estimation algorithm}

Based on corollary \ref{cor:oneopt}, we are now ready to show our trace distance estimation (VTDE) algorithm for arbitrary quantum states $\rho$ and $\sigma$. Fig. \ref{fig:VTDE-diagram} demonstrates the diagram of our algorithm. 

\begin{figure}[ht]
	\centering
	\includegraphics[width=0.48\textwidth]{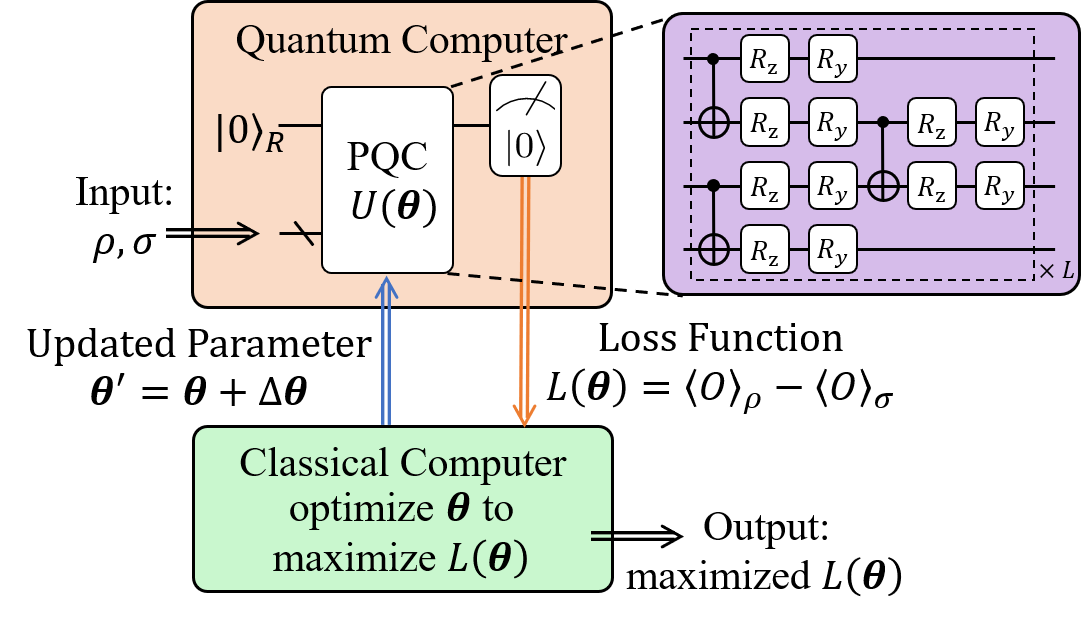}
	\caption{Diagram of VTDE. VTDE employs a single-qubit ancilla for arbitrary number of qubits of $\rho$ and $\sigma$. The outputs $\avg{O}_\rho,\avg{O}_\sigma$ are obtained by locally measuring the overlap with state $\ket{0}_R$. Then the optimization of the measurement outcome is undertaken by a classical computer (the green box). The unitary evolution on the coupled system is simulated by a hardware-efficient ansatz, of which the single qubit rotation gates ($R_y,R_z$) are controlled by classical parameters $\bm\theta$.}
	\label{fig:VTDE-diagram}
\end{figure}

Specifically, with $H_A=\frac{1}{2}(\rho_A-\sigma_A)$ and $\tr H_A=0$, Eq. \eqref{eq:cor3} could be written as 
\begin{equation}
    \begin{aligned}
        D(\rho,\sigma)=&\frac{1}{2}\|\rho_A-\sigma_A\|_1\\
        =&\max_{U} \left(\tr \proj{0}_R \widetilde{\rho}_R-\tr \proj{0}_R \widetilde{\sigma}_R\right),
    \end{aligned}
\end{equation}
where $\widetilde{\rho}_{AR}=U(\rho_A\otimes\proj{r}_R)U^\dagger$, $\rho_R=\tr_A\rho_{AR}$ and $\widetilde{\sigma}_{AR}=U(\sigma_A\otimes\proj{r}_R)U^\dagger$, $\sigma_R=\tr_A\sigma_{AR}$. This observation finally enables us to estimate the trace distance by optimization over unitaries, and the algorithm is given in Algorithm \ref{alg:vtde} in detail. In our VTDE algorithm, $\tr \proj{0}_R \widetilde{\rho}_R$ and $\tr \proj{0}_R \widetilde{\sigma}_R$ are evaluated successively, then the difference between them are maximized by a classical computer. In practice, the ancillary pure state can be initially set as $\ket{r}=\ket{0}$. We adopt a hardware efficient ansatz \cite{Kandala_2017,commeau2020variational} consisting of parameterized single-qubit $R_y$ and $R_z$ rotations, along with CNOT gates on adjacent qubits as entanglement gates (also see Fig.~\ref{fig:VTDE-diagram}) to implement the tunable unitary evolution.

\renewcommand{\algorithmicrequire}{\textbf{Input:}}
\renewcommand{\algorithmicensure}{\textbf{Output:}}
\vspace{0cm}
\begin{algorithm}[H] 
\caption{Variational Trace Distance Estimation (VTDE)}
\label{alg:vtde}
\begin{algorithmic} 
\REQUIRE quantum states $\rho_A$ and $\sigma_A$, circuit ansatz of unitary $U_{AR}(\bm\theta)$, number of iterations ITR;
\ENSURE an estimate of trace distance $D(\rho_A,\sigma_A)$.

\STATE Initialize parameters $\bm\theta$.

\STATE Append $\rho_A$ and $\sigma_A$ with single-qubit state $\ket{0}_R$, respectively.

\FOR{itr $=1,\ldots,$ ITR}
 
\STATE Apply $U_{AR}(\bm\theta)$ to $\rho_A\otimes\proj{0}_R$ and $\sigma_A\otimes\proj{0}_R$, obtain the states $\widetilde{\rho}_{AR}=U_{AR}(\bm\theta)\rho_A\otimes \proj{0}_RU_{AR}(\bm\theta)^\dagger$ and $\widetilde{\sigma}_{AR}=U_{AR}(\bm\theta)\sigma_A\otimes \proj{0}_RU_{AR}(\bm\theta)^\dagger$, respectively.
 
\STATE Evaluate $O_\rho=\tr \proj{0}_R\widetilde{\rho}_{R}$ and $O_\sigma=\tr \proj{0}_R\widetilde{\sigma}_{R}$ by measurement on system $R$. 

\STATE Compute the loss function $\cL_1:=O_\rho-O_\sigma$.
 
\STATE Maximize the loss function $\cL_1$ and update parameters $\bm\theta$.

\ENDFOR

\STATE Output the optimized $\cL_1$ as the trace distance estimate.
\end{algorithmic}
\end{algorithm}

From Algorithm \ref{alg:vtde} one can tell that our VTDE estimates the trace distance for arbitrary quantum states without requiring any pre-knowledge. Furthermore, Algorithm \ref{alg:vtde} and Fig. \ref{fig:VTDE-diagram} imply that VTDE employs a single-qubit ancillary qubit for any input size $n$ of $\rho$ and $\sigma$, and perform local (single-qubit) measurement in each iteration. In this sense, VTDE is general and efficient for NISQ devices. Recall that Algorithm \ref{alg:vtde} can be adapted to trace norm estimation for any $H$ with decomposition $H=\sum_jc^j\rho^j$ based on Corollary \ref{corollary:trace norm}.

Given the VTDE, we briefly discuss analytic gradient and the gradient vanishing (barren plateau \cite{mcclean2018barren}) issues. Analytic gradient enables us to perform gradient descent to optimize the parameters. As our VTDE employs an hardware-efficient PQC, the parameter shift rule \cite{PhysRevA.98.032309,PhysRevA.99.032331} is capable in our case to obtain analytic gradient. As for the barren plateau from which many variational quantum algorithms may suffer. We remark that our VTDE performs single-qubit measurements and takes the result as the loss function. This is essentially equivalent to a local observable, which has been proved to have at worst a polynomially vanishing gradient with a shallow PQC \cite{Cerezoa}. In this sense, our VTDE could avoid the barren plateau issue when the number of layers $L\in\mathcal{O}(\log n)$.

We would like to remark that, the Naimark extension (see Example 2.2 of \cite{Wilde_2013}) together with the property of trace distance that 
\begin{align}
    D(\rho,\sigma)=\max_{0\le P\le 1}\tr P(\rho-\sigma)
    \label{eq:naimark}
\end{align}
also explain VTDE. Indeed our idea is to study the capability of parameterized quantum circuits in extracting spectrum information from local measurements. Nevertheless, VTDE can be generalized for trace norm estimation of arbitrary $H$ (Theorem \ref{thm:tr} and Corollary \ref{corollary:trace norm}), where trace norm estimation in Eq. \eqref{eq:naimark} does not hold if we simply replace $(\rho - \sigma)$ by a non-traceless $H$. This generalization can be lead to wider applications in quantum information like the estimation of log negativity. Our intermediate product (Proposition \ref{prop:AB}) also indicates that the partial sum of the spectrum can be extracted by local measurements.

\subsection{Numerical experiments}
Numerical experiments are undertaken to demonstrate the validity and advantage of VTDE. All simulations including optimization loops are implemented via Paddle Quantum~\cite{Paddlequantum} on the PaddlePaddle Deep Learning Platform~\cite{Paddle,Ma2019}.

We firstly estimate the trace distance between the 4-qubit GHZ state $\ket{\psi}=\frac{1}{\sqrt{2}}(\ket{0000}+\ket{1111})$ and its output after the depolarizing channel:
\begin{align}
    \text{Dep}_p(\rho)=p\tr\rho\frac{I}{2^4}+(1-p)\rho.
\end{align}

In the experiment setting, we let $\update{\rho=\proj{\psi}}$ and $\sigma=\text{Dep}_p(\rho)$, where the channel parameters are $p=\{0.1,0.3,0.5,0.7,0.9\}$. For the hyper-parameters, we set ITR $=120$ and learning rate LR $=0.02$. Fig. \ref{fig:iteration} shows the trace distance learned by VTDE versus the numbers of iterations. As one can see, for all considered channel parameters, the trace distance between $\rho$ and $\sigma$ can be estimated accurately within feasible iterations ($<120$).

\begin{figure}[h]
    \centering
    \includegraphics[width=0.42\textwidth]{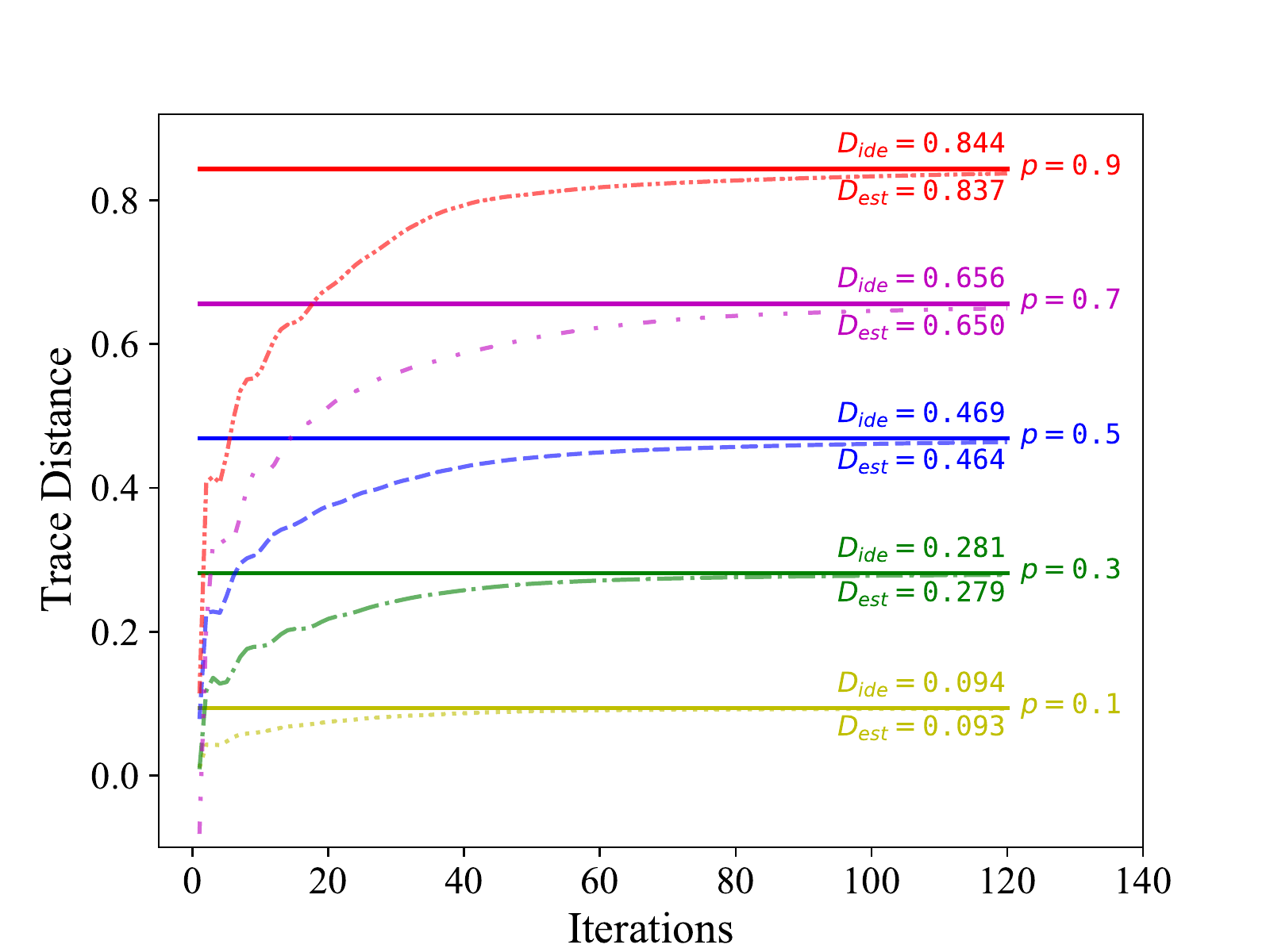}
    \caption{Learning processes of VTDE, where the input states $\rho$ and $\sigma$ are 4-qubit GHZ state and its output after the depolarizing channel. The different colors correspond to considered depolarizing channel parameters $p=\{0.1,0.3,0.5,0.7,0.9\}$, with $p=0.9$ being the top two red lines. We denote the ideal trace distances by $D_{ide}$ and solid lines, and denote the estimated ones by $D_{est}$ and dashed/dotted lines.}
    \label{fig:iteration}
\end{figure}

Next, we explore the required number of layers of the ansatz for input states with distinct ranks. Specifically, we set the number of qubits $n=3$ and Rank($\sigma$) = 2, while Rank($\rho$) ranges from 1 to $2^n=8$. We randomly sample 100 states for each number of ranks of $\sigma$, and compute the accuracy for circuits with 1, 2 and 4 layers. Here the accuracy is defined as $\text{Acc.}=D_{est}/D_{ide}$, where $D_{est,ide}$ are the estimated and ideal trace distances, respectively. Hyper-parameters are the same as the previous experiment. Fig. \ref{fig:rank-layer} summarizes the result, telling that the behavior of a 4-layer circuit is more accurate and stable, while circuits with fewer layers perform badly and unstably for any rank of $\rho$, possibly due to the lack in expressibility \cite{sim2019expressibility}.

\begin{figure}[h]
    \centering
    \includegraphics[width=0.42\textwidth]{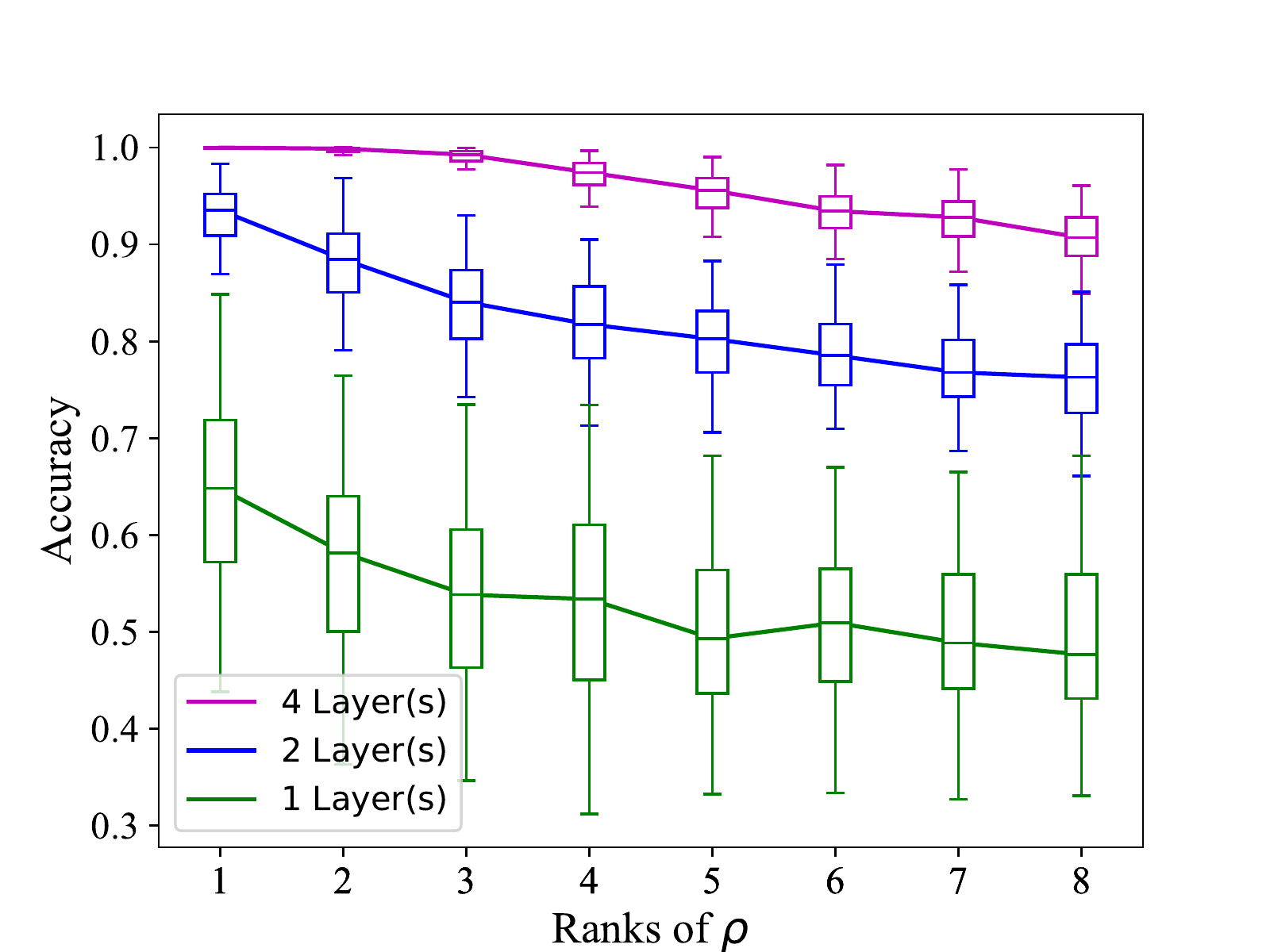}
    \caption{Trace distance learned by VTDE with different number of layers versus the rank of $\rho$. The curves connect the median values in the boxplots displaying the distribution of the estimation of $D$ between 100 sampled state pairs. Each curve corresponds to the 1-, 2-, or 4-layer case.}
    \label{fig:rank-layer}
\end{figure}

\subsection{Implementation on superconducting quantum processor}\label{subsec:VTDE-IBM}

We also apply our VTDE to estimate the trace distance between $\ket{+}$ state and itself affected by a dephasing channel:
\begin{align}
    \text{Deph}_p(\rho)=pZ\rho Z+(1-p)\rho.
\end{align}
The result with comparison to the values achieved from simulation are listed in Table \ref{table:VTDE-IBM}. The experiments are performed on the \href{https://quantum-computing.ibm.com/}{\textit{IBM quantum}} platform, loading the quantum device $ibmq\_quito$ containing 5 qubits.

\begin{table}[h]
\small
\centering
\begin{tabular}{|c ||c| c|c|}
 \hline
 Backends & mean & variance & error rate\\ [0.5ex] 
 \hline\hline
 Theoretical value & $0.7$ & - & - \\ 
 Simulator & $0.70322$ & $0.00010$ & 0.46\%\\
 $ibmq\_quito$ & $0.64312$ & $0.00025$ & 8.13\%\\
 \hline
 \end{tabular}
\caption{Error analysis for VTDE between $\rho=\ket{+}$ and $\sigma=\text{Deph}_p(\rho)$. We choose $p=0.7$ and repeat the experiment 10 times independently.}
\label{table:VTDE-IBM}
\end{table} 

The input state we choose here are $\rho=\ket{+}$ and $\sigma=\text{Deph}_{0.7}(\rho)$. On the quantum device, the input state $\sigma$ is prepared by applying two $Ry$ gates on the working and ancilla qubits with parameters $\pi/2$ and $2\arcsin{\sqrt{0.7}}\approx1.982$ respectively, followed by a Control-$Z$ gate, then discarding the ancilla. Next, we operate parameterized quantum circuit on qubits, and use sequential minimal optimization \cite{Nakanishi2020Sequential} to optimize the parameters until the cost converge to its maximum. We repeat 10 independent experiments with same input states and randomly initialized parameters.

As demonstrated in Table \ref{table:VTDE-IBM}, the estimates converges stably for the simulator and the quantum device. Correspondingly, the achieved trace distance estimates by simulator (0.7032) is closer to the theoretical values (0.7) than that of quantum device (0.64312), which may caused by the quantum device noises.

\section{Variational fidelity estimation}\label{sec:fidelity}
In this section, we introduce Variational Fidelity Estimation (VFE) as a hybrid quantum-classical algorithm for estimating fidelity in the most general case where two mixed states are provided. The intuition of VFE lies in Uhlmann's theorem and the freedom in purification, based on which we prove that optimization over unitaries can obtain the fidelity. The purification of each quantum state is required by the optimization, for which we design a variational quantum state learning (VQSL) algorithm \update{as a subroutine}. Both VFE and VQSL employs PQC to implement the optimization over unitaries. We give the theory behind and the process of VFE and VQSL, respectively, \update{followed by experimental results from a simulator and a 5-qubit IBMQ superconducting processor.} The \update{workflow} of VFE is shown in Fig. \ref{fig:VFE-diagram}.

\subsection{Theory and algorithm for fidelity estimation}

\begin{figure}[ht]
	\centering
		\includegraphics[width=0.48\textwidth]{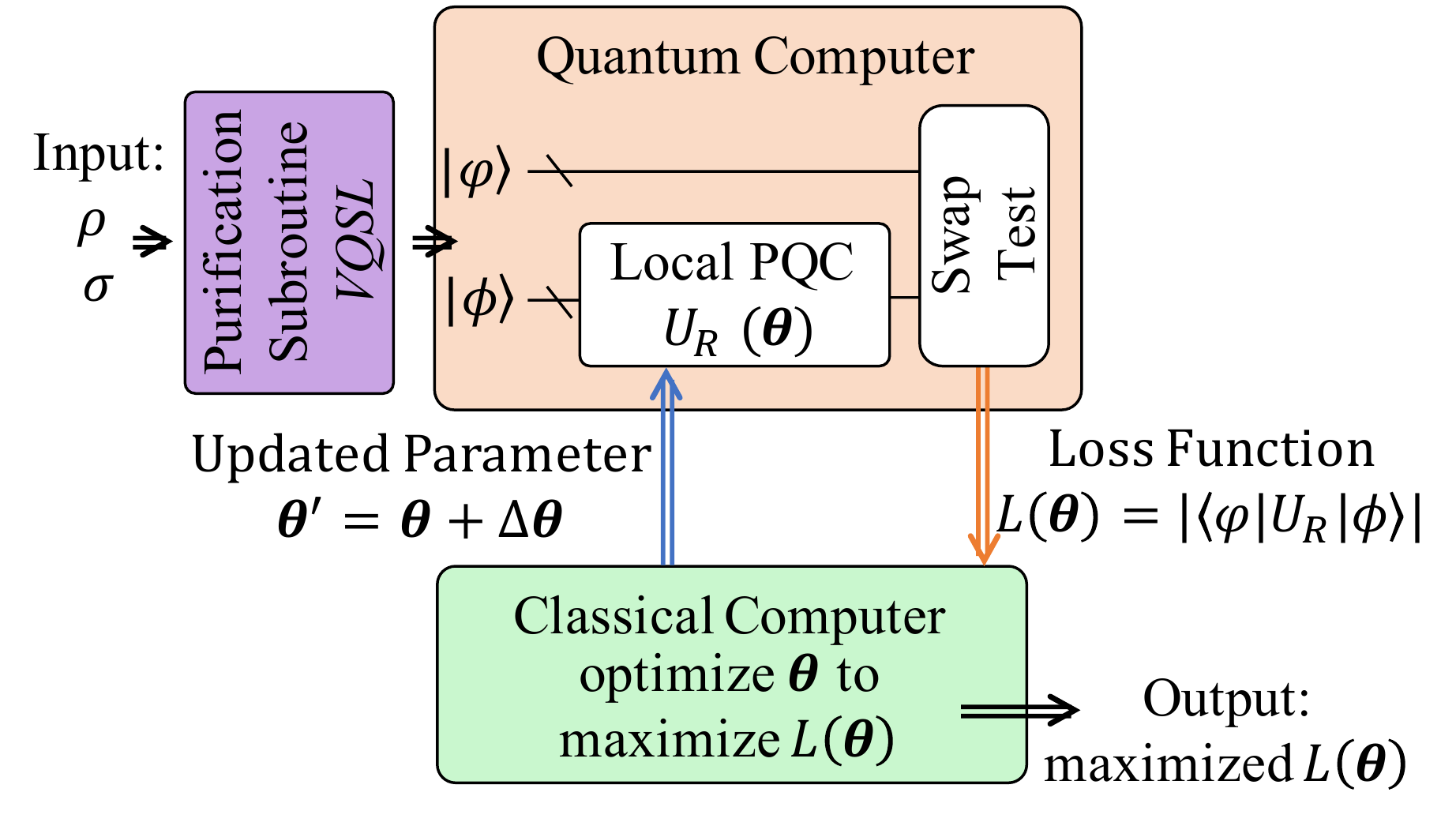}
	\caption{\update{Workflow} of VFE. VFE employs a purification subroutine to purify $\rho$ and $\sigma$. The outputs $\bra{\psi}U_R\ket{\phi}$ are obtained by the swap test. Then the optimization of the measurement outcome is undertaken by a classical computer (the green box). The unitary evolution on the ancillary system is \update{realized} by a hardware efficient ansatz.}
	\label{fig:VFE-diagram}
\end{figure}

Suppose we have two unknown quantum states, $\rho_A$ and $\sigma_A$, on system $A$. With arbitrary purification of $\rho_A$ and $\sigma_A$ (denoted by $\ket{\psi}_{AR}$ and $\ket{\phi}_{AR}$, respectively), Corollary \ref{cor:fid} allows us to estimate the fidelity of $\rho_A$ and $\sigma_A$. As the prerequisite of Corollary \ref{cor:fid}, we first introduce Uhlmann's theorem and the freedom in purification:

\begin{theorem}[Uhlmann's theorem, \cite{Nielsen2010}]
Suppose $\rho$ and $\sigma$ are states of quantum system $A$. Then, for given quantum purification $\ket\psi_{AR}$ of $\rho_A$ on system $AR$, 
\begin{align}
F(\rho_A,\sigma_A) = \max_{\ket{\phi}_{AR}} |\braket{\psi}{\phi}_{AR}|, 
\end{align}
where the maximization is over all purifications $\ket\phi_{AR}$ of $\sigma_{A}$ on system $AR$.
\label{thm:uhl}
\end{theorem}

\begin{lemma}[Freedom in purification, \cite{Nielsen2010}]
For two purifications $\ket{\phi_1}_{AR},\ket{\phi_2}_{AR}$ of $\sigma_A$ on system $AR$, there exists a unitary transformation $U_R$ such that
\begin{align}
\ket{\phi_2}_{AR} = (I_A\otimes U_R)\ket{\phi_1}_{AR},
\end{align}
\label{lem:fre}
\end{lemma}

Based on the fundamental properties of quantum fidelity discussed above, we observe that the following optimization problem characterizes the fidelity function.
\begin{corollary}
For any quantum states $\rho_A$ and $\sigma_A$ on system $A$, and arbitrary purification $\ket{\psi}_{AR}$ and $\ket{\phi}_{AR}$ of $\rho_A$ and $\sigma_A$, it holds that
\begin{align}
    F(\rho,\sigma) = \max_{U_R} |\bra\psi_{AR} (I_A\otimes U_R)\ket\phi_{AR}|.
\end{align}
where the optimization is over any unitaries on system $R$.
\label{cor:fid}
\end{corollary}

Corollary \ref{cor:fid} gives us an elegant variational representation of the fidelity function that only requires purification of input states and optimization over the ancillary system. Following this line of reasoning, we design a variational quantum algorithm to estimate the fidelity of quantum states $\rho$ and $\sigma$. We note that in Algorithm \ref{alg:vfe} the procedure of purifying $\rho$ and $\sigma$ is realized by VQSL as a subroutine, which will be presented in Sec.~\ref{sec:sub}.
Our method formulates the problem of directly calculating the fidelity between two mixed states into an optimization procedure over the ancillary system of two purified states. At the cost of extra subroutines, this approach could deal with arbitrary high-rank states as long as we provide enough ancillary qubits. For numerical performance analysis, we refer to Sec.~\ref{sec:vfe-exp} and Sec.~\ref{sec:vfe-ibmq}. Besides, we also derive the analytical gradient of this loss function in Appendix \ref{proof:gradient_VFE}, \update{which is required for gradient-based classical optimization.}

\vspace{0cm}
\begin{algorithm}[t] 
\caption{Variational Fidelity Estimation (VFE)}
\begin{algorithmic} \label{alg:vfe}
\REQUIRE quantum states $\rho_A$ and $\sigma_A$, circuit ansatz of unitary $U_{R}(\bm\theta)$, number of iterations ITR;
\ENSURE an estimate of fidelity $F(\rho_A,\sigma_A)$.

\STATE Use Purification Subroutine to learn the 
purified state $\ket{\psi}_{AR}$ of $\rho_A$ and the 
purified state $\ket{\phi}_{AR}$ of $\sigma_A$.

Initialize parameters $\bm \theta$.
 
\FOR{itr $=1,\ldots,$ ITR}

\STATE Apply $U_R(\bm\theta)$ to $\ket{\phi}_{AR}$ and obtain the resulting state $\ket{\widetilde\phi}_{AR} = I_A \otimes U_R(\bm\theta)\ket{\phi}_{AR}$.
 
\STATE Compute the loss function $\cL_2(\bm\theta):=|\braket{\widetilde\phi}{\psi}|_{AR}$.
 
\STATE Maximize the loss function and update parameters $\bm\theta$;

\ENDFOR
 
\STATE Output the optimized  $\cL_2$ as the final fidelity estimation;
\end{algorithmic}
\end{algorithm}

\subsection{Subroutine - quantum state learning and purification}
\label{sec:sub}
The task of quantum state learning is to find the correct unitary operation such that one could prepare any target mixed state $\rho$ from an initialized state (usually $|0\rangle\langle0|$ on each qubit). In Ref.~\cite{Lee2018a}, the method of state learning for a pure state $\rho$, where $\text{Rank}(\rho)=1$, is proposed. Our approach further generalizes it to work for a mixed state $\rho$, where $\text{Rank}(\rho)\geq 1$. In order to implement our algorithm in NISQ devices, we consider quantum state learning via PQC and the framework of variational quantum algorithm (VQA). The main issue is how to design a faithful loss function $\cL_3$ to guide the learning direction. Such a loss function should be able to quantify the closeness between target $\rho$ and prepared state $\chi$. For pure state cases, simply maximize the state overlap $\tr{\rho \chi}$ could help us learn the target state. But for general mixed states, this is not the case. Consider a counter example of $\rho = I/2$ and $\chi$ being a random one-qubit mixed state. The overlap is always $\tr{\rho \chi} = 1/2$ but this is clearly not the correct distance measure between $\rho$ and $\chi$. On the other hand, the Hilbert-Schmidt norm defined as follows could be a good candidate:
\begin{equation}
\Delta(\rho, \chi):= \|\rho-\chi\|_2^2 =\tr(\rho-\chi)^2.
\label{eq:HS-norm}
\end{equation}
As the unknown state $\rho$ is fixed, we have

\begin{align}
&\argmin_{\chi} \Delta(\rho,\chi) \notag \\
=& \argmin_{\chi} \tr \rho^2 + \tr\chi^2 - 2\tr\rho\chi  \notag \\
=& \argmin_{\chi} \tr\chi^2 - 2\tr\rho\chi,
\end{align} 
where $0<\tr\rho^2<1$ for mixed state $\rho$ is a constant and hence doesn't influence the optimization direction. 
Therefore, we choose our loss function to be
\begin{align}
\min_{\bm\t} \cL_3(\bm\theta):=\tr\chi^2 - 2\tr\rho\chi.
\label{eq:purification-loss}
\end{align}

Note that this loss function can be implemented on near-term devices since the state overlap can be computed via the Swap test~\cite{Buhrman2001,Gottesman2001a}. As a brief reminder, the Swap test evaluates the overlap of two arbitrary states by a single qubit measurement after a combination of Hadamard gates and controlled-swap gates. Evidence has been found that the Swap test has a simple physical implementation in quantum optics~\cite{Ekert2002,Garcia-Escartin2013} and can be experimentally implemented on near-term quantum hardware~\cite{Islam2015,Patel2016,Linke2018}.

After discussing the general picture of state learning, we introduce its application in learning the purification of quantum state $\rho_A$ on system $A$. This can be done by providing an ancillary qubit system $R$ with dimension $d_R$ and initializing the complete system with $\ket{\chi_0} =\ket{00}_{AR}$. Then we apply a parametrized unitary operation $U_{AR}(\bm \theta)$ to drive the system and use classical optimization methods to minimize Eq.\eqref{eq:purification-loss}. In the case of learning purification, we need to set
\begin{align}
\update{\chi_A = \tr_R \big[U_{AR}(\bm\theta)\proj{\chi_0}U_{AR}^{\dagger} (\bm\theta)\big],}
\end{align}
where the symbol $\tr_R$ denotes the partial trace operation with respect to the ancillary system $R$. With the above set up, we introduce a variational quantum algorithm to learn the purification of a quantum state $\rho_A$ as follows.

\begin{algorithm}[H] 
\caption{Variational Quantum State Learning (VQSL)}
\begin{algorithmic} \label{alg:vqsl}
\REQUIRE quantum states $\rho_A$, circuit ansatz of unitary $U_{AR}(\bm\theta)$, number of iterations ITR;
\ENSURE a purification of $\rho_A$.

\STATE Initialize parameters $\bm\theta$.

\FOR{itr $=1,\ldots,$ ITR}

\STATE Apply $U_{AR}(\bm\theta)$ to three equivalent initial states $\ket{\chi_0}$ on system $AR$ and obtain the resulting partial states: $\chi^1_{A} =\chi^2_{A}=\chi^3_{A} = \tr_R \big[U_{AR}(\bm\theta)\proj{\chi_0}U_{AR}^{\dagger}(\bm\theta) \big]$;  
 
\STATE Measure the overlap $\avg{O}_1 = \tr (\chi^1_{A}\chi^2_{A})$ via Swap Test;

\STATE Measure the overlap $\avg{O}_2=\tr (\rho_{A} \chi^3_{A})$ via Swap Test;  
\STATE Compute the loss function $\cL_3(\bm\theta) =  \avg{O}_1 -2\avg{O}_2 $;

\STATE Perform optimization for $\cL_3(\bm\theta)$ and update parameters $\bm\theta$;

\ENDFOR
\STATE Output the final purified state $\ket{\psi}_{AR} = U_{AR}(\bm\theta^*)\ket{00}_{AR}$;
\end{algorithmic}
\end{algorithm}

We want to emphasize that the scope of state learning is much broader than the purification learning task. One could further develop this approach for quantum state preparation and many other applications. For our purpose, learning a purified state to estimate the fidelity is enough. At the cost of introducing ancillary qubits, one could prepare a purification $\ket{\psi}_{AR}$ of target mixed state $\rho_A$ with high fidelity. Then, it will be important to study the performance of VQSL with different ancillary dimensions $d_R$. The following analysis suggests that only a few ancillary qubits are necessary for low-rank states. For best performance, one should choose the dimension of the ancillary system to be the same as the original system $A$ such that $d_R = d_A$.

\begin{proposition}
Suppose the input state $\rho_{A}$ has the spectral decomposition $\rho_{A} = \sum_{j=1}^k \lambda_j\proj{\psi_j}$ with decreasing spectrum $\{\lambda_j\}_{j=1}^k$. \update{There exists a quantum circuit $U_{AR}$ that generates $\chi_A$ from $\ket{00}_{AR}$ to approximate the target state $\rho_A$} with a fidelity that satisfies

\begin{equation}
  F(\rho_A,\chi_A) \!=\! \left\{ \!\begin{array}{ll} 1 & \text{if $k \le d_R $}, \\ 
  \sqrt{\sum_{j=1}^{d_R}\lambda_j} & \text{otherwise,} \end{array} \right. 
\label{eq:vqsl_bound}  
\end{equation}
where the second case is at least  $\sqrt{{d_R}/{k}}$.
\end{proposition}
\begin{proof}
Assume $d_R\le d_A$. the goal is to solve 
\begin{align}
\max_{U_{AR}} F(\rho_A,\chi_A),
\end{align} 
\update{where $\chi_A = \tr_R \big[U_{AR}\proj{00}_{AR}U^\dagger_{AR}\big]$ is the reduced state on the subsystem $A$.} For convenience, denote the state $U_{AR}\ket{00}_{AR}$ by $\ket{\psi}$ and hence
\begin{align}
\chi_A = \tr_R \proj{\psi}_{AR}.
\end{align} 
Since the Schmidt rank of $\ket{\psi}$ is at most $d_R$, the rank of $\chi_A$ is also at most $d_R$. 

If $k\le d_R$, one could choose unitary $U_{AR}$ such that $\ket{\psi} = U_{AR}\ket{00}_{AR} = \sum_{j=1}^k \sqrt{\lambda_j}\ket{\psi_j}_A\ket{\varphi_j}_R$, where $\{\ket{\psi_j}_A\}$ and $\{\ket{\varphi_j}_R\}$ are orthonormal state sets. This ensures that $\chi_A = \tr_R \proj{\psi}= \sum_j \lambda_j\proj{\psi_j} = \rho_A$ and leads to a fidelity $F(\rho_A, \chi_A) = 1$.
On the other hand, if $k>d_R$, there exists an unitary $U$ such that $\ket{\psi} = U_{AR}\ket{00}_{AR} = \sum_{j=1}^{d_R} \sqrt{\xi_j}\ket{\psi_j}_A\ket{\varphi_j}_R$.
In particular, one could further choose $\xi_j = \frac{\lambda_j}{\eta}$ for $j=1,\cdots,d_R$, where the denominator is defined as $\eta = {\sum_{\ell=1}^{d_R}\lambda_\ell}$.
Then, one can calculate the fidelity $F$ as
\begin{align}
    F(\rho_A, \chi_A) & = \tr \sqrt{ \chi_A^{1/2}\rho_A \chi_A^{1/2}} \notag \\
    & = \sum_{j=1}^{d_R}\sqrt{\lambda_j^2/\eta} \notag \\
    & = \sum_{j=1}^{d_R}\lambda_j/\sqrt{\eta} \notag \\
    & = \sqrt{\eta}  \ge \sqrt{\frac{d_R}{k}}.
\end{align}
\end{proof}

\subsection{Numerical experiments and Trainability}
\label{sec:vfe-exp}

In this section, we conduct simulations to investigate the performance of VFE for state fidelity estimation and its subroutine VQSL for purification learning. The parametrized quantum circuit $U_{AR}(\bm \t)$ used for Algorithm \ref{alg:vqsl} VQSL and $U_{R}(\bm \t)$ for Algorithm \ref{alg:vfe} are recorded in Appendix \ref{cir:VFE}. \update{All simulations including optimization loops are implemented via Paddle Quantum~\cite{Paddlequantum} on the PaddlePaddle Deep Learning Platform~\cite{Paddle,Ma2019}.}

We firstly generate 10 random pairs of full-rank density matrices $\{\rho_A^{(1)},\sigma_A^{(1)}|\cdots |\rho_A^{(10)},\sigma_A^{(10)} \}$ for $n_A = \{1,2,3\}$ \update{qubits} and calculate the deviation $\Delta F$ between the target fidelity and estimated fidelity. We choose $d_R = d_A$ for best performance in purification learning. The results are summarized in Table.~\ref{table:fid-random}. The average error rate can be reduced to the level of $<0.5\%$. \update{The maximum average error $\mathbb{E}_{max}[\Delta F] \approx 0.2476\%$ happens at $n_A = 2$, and no clear scaling phenomenon of average error is observed.}

\begin{table}[h]
\small
\centering
\begin{tabular}{|c ||c| c| c|}
 \hline
 Qubit number $\#$ & $n_A = 1$ & $n_A = 2$ & $n_A = 3$ \\ [0.5ex] 
 \hline\hline
 Average error rate $\mathbb{E}[\Delta F]$ & \quad $0.1694\%$ \quad&\quad $0.2476\%$ \quad&\quad $0.1960\%$ \quad\\ 
 Standard deviation $\sigma[\Delta F]$ & \quad $0.1594\%$ \quad&\quad $0.1457\%$ \quad&\quad $0.1388\%$ \quad\\
 \hline
 \end{tabular}
\caption{Error analysis for fidelity estimation VFE between randomly generated full-rank density operators. For purification learning, the Adam optimizer~\cite{kingma2014adam} is adopted and hyper-parameters are taken to be depth $L=6$, learning rate LR $= 0.2$, and iteration loops  ITR $= 100$ (the latter two \update{hyper-parameters} are same for VFE).}
\label{table:fid-random}
\end{table} 

\begin{figure}[h]
    \centering
    \includegraphics[width=0.5\textwidth]{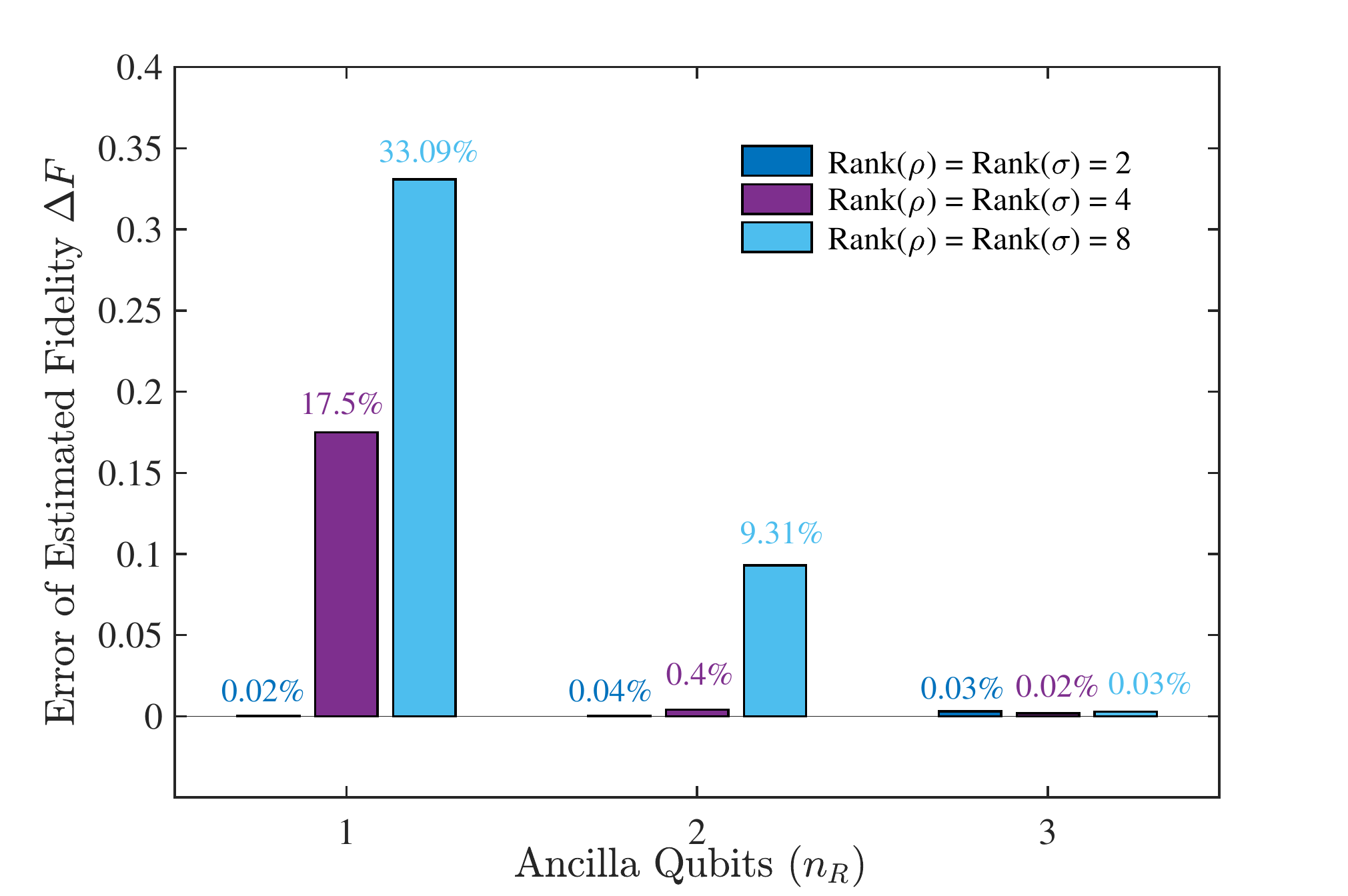}
    \caption{\update{Performance of VFE given two randomly generated 3-qubit mixed states $\rho$, $\sigma$ with different number of ancilla qubits $n_R$. 
    }}
    \label{fig:fid-experiment}
\end{figure}

Next, we conduct simulations to qualitatively study the influence of limited ancillary qubits given high-rank states $k>d_R$. 3 random pairs of density matrices, $\{\rho_A^{(1)},\sigma_A^{(1)}|\rho_A^{(2)},\sigma_A^{(2)}|\rho_A^{(3)},\sigma_A^{(3)}\}$, are prepared with Rank$=\{2,4,8\}$ respectively. Then, we test the performance of VFE with different number of ancillary qubits $n_R = \{1,2,3\}$. The results are summarized in Fig.\ref{fig:fid-experiment}. 
\update{By providing enough ancilla qubits, VFE could estimate the fidelity between two arbitrary full rank mixed states. If not, the purification subroutine will be restricted to produce pure states with fidelity bounded by Eq.~\eqref{eq:vqsl_bound}, which is consistent with the numerical simulation.}
\update{This scalability requirement offers the flexibility of providing few ancilla qubits for low-rank states but can bring challenges when accurately evaluating the fidelity between two high-rank quantum state in large dimensions. } 

\update{Here we discuss the trainability issues for VFE and possible solutions.
With an extended system dimension introduced by the purification, VFE might exhibit a barren plateau~\cite{mcclean2018barren} due to {a global loss function}. This would result an exponentially suppressed gradient with respect to the problem dimension, a phenomenon known as the barren plateau (BP). Ref.~\cite{arrasmith2020effect} shows BP is independent of the optimization methods, meaning that simply change into a gradient-free optimizer would not help. In order to mitigate this trainability issue, several approaches have been proposed such as the parameter initialization and correlation strategies \cite{Grant2019initialization, volkoff2021large},  layer-wise learning \cite{skolik2020layerwise}, and variable  structure  ansatzes~\cite{grimsley2019adaptive, rattew2019domain, bilkis2021semi}. We leave the adaptation of these strategies to VFE for future study.}

\subsection{Implementation on superconducting quantum processor}
\label{sec:vfe-ibmq}
We also apply our VFE to estimate the fidelity between states which are $\ket{+}$ states affected by dephasing channel $\text{Deph}_p$ with different intensity $p$. The result with comparison to the values achieved from simulation are listed in Table \ref{table:VFE-IBM}. The experiments are performed on the \href{https://quantum-computing.ibm.com/}{\textit{IBM quantum}} platform, loading the quantum device $ibmq\_quito$ containing 5 qubits.

\begin{table}[h]
\small
\centering
\begin{tabular}{|c ||c| c|c|}
 \hline
 Backends & mean & variance & error rate\\ [0.5ex] 
 \hline\hline
 Theoretical value & $0.70710$ & - & - \\ 
 Simulator & $0.70721$ & $0.00002$ & 0.016\%\\
 $ibmq\_quito$ & $0.71692$ & $0.00003$ & 1.388\%\\
 \hline
 \end{tabular}
\caption{Error analysis for VFE between $\rho = \text{Deph}_{p_1}(\ket{+})$ and $\sigma = \text{Deph}_{p_2}(\ket{+})$. We choose $p_1=0.2$ and $p_2=0.9$, and repeat the experiment 10 times independently.}
\label{table:VFE-IBM}
\end{table} 

The input state we choose here are $\rho=\text{Deph}_{p_1}(\ket{+})$ and $\sigma=\text{Deph}_{p_2}(\ket{+}$), where $p_1=0.2$ and $p_2=0.9$. On the quantum device, the input states are prepared similarly to the preparation in Sec. \ref{subsec:VTDE-IBM} and we keep the ancilla as a perification. Next, we operate parameterized quantum circuit on the ancilla, and use sequential minimal optimization \cite{Nakanishi2020Sequential} to optimize the parameters until the loss function converges to its maximum. We repeat 10 independent experiments with same input states and randomly initialized parameters. As demonstrated in Table \ref{table:VFE-IBM}, the estimated fidelity from our method are close to the ideal value with stable performance due to a very small variance.

\section{Conclusion and outlook}\label{sec:conclusion}
In this work, we have introduced near-term quantum algorithms VTDE and VFE to estimate trace distance and quantum fidelity. A strength of our algorithms is that they estimate the metrics directly rather than estimating their bounds. Our algorithms also do not require any assumption on the unknown input states. These algorithms are executable on near-term quantum devices equipped with parameterized quantum circuits. In particular, VTDE could be easily generalized for trace norm estimation for any Hermitian matrix and could avoid the barren plateau issue with logarithmic depth parameterized circuits. 

Beyond benchmarking the quantum algorithms' behavior, our VTDE and VFE could have a wide range of applications in quantum information processing. A direct extension of VTDE might be the estimation for the diamond norm~\cite{Kitaev2000,Watrous2009,watrous2012simpler}, which is a widely-used distance measure for quantum channels. The trace distance can also be applied to quantify the Bell non-locality \cite{Brito_2018}, quantum entanglement \cite{PhysRevA.65.032314}, and the security of quantum cryptography protocols \cite{yuen2014trace}, where VTDE could be utilized as a practical subroutine. VFE may be applied to evaluate the conditional quantum mutual information of tripartite quantum states~\cite{Berta_2016}. It is also of great interest to have a further study on the estimation of sandwiched/Geometric R\'enyi relative entropies~\cite{Muller-Lennert2013,Wilde2014a,Matsumoto2018} 
as an extension of VFE. Efficient estimations of these distance measures could be used to further evaluate the resource measures of entanglement, coherence, magic and other resources in quantum information~\cite{Chitambar2018,Plenio2007,Wang2016d,Vedral2002,Wang2020c,Datta2009,WWW19,Yuan2019,Streltsov2016,Wang2018,Veitch2014}.

\section*{Acknowledgements.}
We thank Runyao Duan, Yuao Chen, and Mark M. Wilde for helpful discussions. 
R. C. and Z. S. contributed equally to this work.
This work was done when R. C., Z. S., and X. Z. were research interns at Baidu Research.


\vspace{2cm}
\onecolumngrid
\vspace{2cm}
\begin{center}
{\textbf{\large Supplemental Material of Variational Quantum Metric Estimation}}
\end{center}

\renewcommand{\theequation}{S\arabic{equation}}
\renewcommand{\theproposition}{S\arabic{proposition}}
\setcounter{equation}{0}
\setcounter{table}{0}
\setcounter{section}{0}
\setcounter{proposition}{0}

\section{Detailed proofs}\label{sec_appendix}
\subsection{Proof of Proposition~\ref{prop:H sum of some eig}}\label{proof:H sum of some eig}
\renewcommand\theproposition{\ref{prop:H sum of some eig}}
\setcounter{proposition}{\arabic{proposition}-1}
\begin{proposition}
For any Hermitian matrix $H_{AB}\in\cL{(\mathcal{H}_A\otimes\mathcal{H}_B)}$ with spectral decomposition as Eq.~\eqref{eq:spedec}, respectively denote the dimension of $\mathcal{H}_A$, $\mathcal{H}_B$ by $d_A$, $d_B$. It holds that
\begin{equation}
\max_{U} \tr \proj{0}_A \widetilde{H}_A =\sum_{j=1}^{d_B}h^j_{AB},
\end{equation}
where $\widetilde{H}_A=\tr_B \widetilde{H}_{AB}$, $\widetilde{H}_{AB} = UH_{AB}U^\dagger$, and the optimization is over all unitaries.
\end{proposition}

\begin{proof}
By the definition of partial trace, we have
\begin{equation}
\begin{aligned}
    \tr \proj{0}_A \widetilde{H}_A=&\tr \proj{0}_A\otimes I_B \widetilde{H}_{AB}\\
    =&\sum_{j=1}^{d_B}\widetilde{H}_{AB}^{j,j},
\end{aligned}
\end{equation}
where $\widetilde{H}_{AB}^{j,j}$ is the $j$-th diagonal element of the matrix representation of $\widetilde{H}_{AB}$. Since $H_{AB}$ and $\widetilde{H}_{AB}$ have the same decreasing spectrum $\{h^j_{AB}\}_{j=1}^{d_Ad_B}$, which majorizes the diagonal elements of $\widetilde{H}_{AB}$ \cite[p. 254, Theorem 4.3.50]{horn2012matrix}, we have
\begin{equation}
\begin{aligned}
    \tr \proj{0}_A \widetilde{H}_A=&\sum_{k=1}^{d_B}\widetilde{H}_{AB}^{k,k}\\
    \le&\sum_{j=1}^{d_B}h^j_{AB}.
\end{aligned}
\label{eq:maj}
\end{equation}

As the equality of Eq. \ref{eq:maj} holds when $\widetilde{H}_{AB}$ is diagonolized under the computational basis, the proof is complete.
\end{proof}

\subsection{Proof of Theorem \ref{thm:tr}}
\label{proof:tr}
\renewcommand\theproposition{\ref{thm:tr}}
\setcounter{proposition}{\arabic{proposition}-1}
\begin{theorem}
For any Hermitian $H_A$ on $n$-qubit system $A$, and any single-qubit pure state $\ket{r}$ on system $R$, it holds that
\begin{align}
\|H_A\|_1 = \max_{U^+} \tr \proj{0}_R Q_R^++\max_{U^-}\tr \proj{0}_R Q_R^-,
\end{align}
where $Q^\pm_R = \tr_AQ^\pm_{AR}$, $Q_{AR}^\pm = U^\pm (\pm H_A \otimes \proj{r}_R) U^{\pm\dagger}$, and each optimization is over unitaries on system $AR$.
\end{theorem}

\begin{proof}
\update{The main idea of the proof of Theorem \ref{thm:tr} is to show that the two maximization item approximates the sum of the positive and negative eigenvalues of $H_A$ respectively. Note that the trace norm $\|H_{A}\|_1$ equals to the sum of the absolute values of $H_{A}$'s eigenvalues:}
\begin{equation}
    \|H_{A}\|_1=\sum_{j=1}^{d_A}|h_A^j|=\sum_{h_A^j>0}h_A^j+\sum_{h_A^j<0}-h_A^j.
    \label{eq:trasum}
\end{equation}

\update{Thus, showing that $\max_{U^+} \tr \proj{0}_R Q^+_R=\sum_{h_A^j>0}h_A^j$ and $\max_{U^-} \tr \proj{0}_R Q^-_R=\sum_{h_A^j<0}-h_A^j$ will complete the proof.}

We first prove that $\max_{U^+} \tr \proj{0}_R Q^+_R=\sum_{h_A^j>0}h_A^j$. Observe that $H_A$ and $H_A\otimes \proj{r}_R$ have the same spectrum, up to $d_A$ 0 eigenvalues. Denote the decreasing spectrum of $H_A\otimes\proj{r}_R$ by $\{h'^j_{AR}\}_{j=1}^{2d_A}$. Then this observation implies that
\begin{equation}
\begin{aligned}
    \sum_{j=1}^{d_A}h'^j_{AR}=\sum_{h^j_{A}>0}h^j_{A}.
\end{aligned}
\label{eq:halsum}
\end{equation}

According to Proposition \ref{prop:AB}, for $Q_{AR}^+ = U^+(H_A\otimes \proj{r}_R)U^{+\dagger}$, $\max_{U^+} \tr \proj{r}_R Q^+_R$ equals to the sum of the first half eigenvalues of $H_A\otimes \proj{r}_R$, therefore
\begin{equation}
    \max_{U^+} \tr \proj{0}_R Q^+_R=\sum_{j=1}^{d_A}h'^j_{AR}=\sum_{h_A^j>0}h_A^j.
    \label{eq:meshalplu}
\end{equation}

Similarly, for $Q_{AR}^- = U^-(- H_A\otimes \proj{r}_R)U^{-\dagger}$, $\max_{U^-} \tr \proj{r}_R Q^-_R$ equals to the sum of the first half eigenvalues of $-H_A\otimes \proj{r}_R$, or to say, the sum of the absolute value of the last half eigenvalues of $H_A\otimes \proj{r}_R$:
\begin{equation}
    \max_{U^-} \tr \proj{0}_R Q^-_R=\sum_{j=d_A+1}^{2d_A}-h'^j_{AR}=\sum_{h_A^j<0}-h_A^j.
    \label{eq:meshalmin}
\end{equation}

Together with Eq. \eqref{eq:trasum}, \eqref{eq:meshalplu}, and \eqref{eq:meshalmin} we have
\begin{align}
\|H_A\|_1 = \max_{U^+} \tr \proj{0}_R Q_R^++\max_{U^-}\tr \proj{0}_R Q_R^-,
\end{align}
which completes the proof.
\end{proof}

\section{A Naive algorithm for trace distance estimation}
From the definition of trace norm, one may naturally try to learn $\|H\|_1$ by twice optimizing over projectors. Explicitly, we have the following proposition:
\renewcommand{\theproposition}{S\arabic{proposition}}
\begin{proposition}
\label{prop:proj}
The trace norm $\|H\|_1$ can be obtained by:
\begin{equation}
\begin{aligned}
    \|H\|_1 = \max_{P_+}\mathrm{Tr}(P_+H) + \max_{P_-}\mathrm{Tr}(-P_-H),
\end{aligned}
\label{eq:nai}
\end{equation}
where each optimization is over projectors. The optimized projector $P_\pm$ is rank-$k_\pm$, where $k_\pm$ is the number of the {positive/negative} eigenvalues of $H$.
\end{proposition}

\begin{proof}
Consider the spectrum decomposition of $H$:
\begin{equation}
\begin{aligned}
    H = &\sum_{j=1}^{d}h^j\proj{\psi_j}\\
    =&\sum_{h^j>0}h^j\proj{\psi_j}+\sum_{h^j<0}h^j\proj{\psi_j}
\end{aligned}
\end{equation}
where $d$ is the dimension of $H$. Let
$$S=\sum_{h^j>0}h^j\proj{\psi_j},~T=-\sum_{h^j<0}h^j\proj{\psi_j}$$
so that 
\begin{equation}
    H=S-T,~\|H\|_1=\tr S+\tr T,
    \label{eq:ST}
\end{equation}
and $S$ and $T$ are positive and orthogonal. In this sense, we have
\begin{equation}
\begin{aligned}
    &\tr (P_+H)=\tr (P_+(S-T))\le\tr (P_+S)\le\tr S\\
    &\tr (-P_-H)=\tr (P_-(T-S))\le\tr (P_-T)\le\tr T
\end{aligned}
\label{eq:Ppm}
\end{equation}

The two inequalities become equal when $P_\pm$ is the projector onto the support of $S$ and $T$, respectively. Thus the optimized projectors $P_\pm$ is rank-$k_\pm$, where $k_\pm$ is the rank of $S/T$, i. e., the number of the {positive/negative} eigenvalues of $H$. Eq. \ref{eq:ST} and Eq. \ref{eq:Ppm} imply Eq. \ref{eq:nai}, which finish the proof.
\end{proof}

Using similar technique in the proof of Corollary \ref{cor:oneopt} we can have following corollary on trace distance estimation from Proposition \ref{prop:proj}.
\begin{corollary}
For any quantum states $\rho$ and $\sigma$, it holds that
\begin{equation}
    D(\rho,\sigma)=\max_P\tr P(\rho-\sigma),
\end{equation}
and the optimized projector $P$ is rank-$k$, where $k$ is the number of the positive eigenvalues of $\rho-\sigma$.
\end{corollary}

In practice, any rank-$k$ projector can be characterized by 
\begin{equation}
    P_{\mathrm{rank}-k}=\sum_{j=1}^{k}U\proj{j}U^\dagger,
\end{equation}
and one can obtain $\tr P_{\mathrm{rank}-k}H=\sum_{j=1}^{k}\tr HU\proj{j}U^\dagger$.

If one does not have $k$ as prior knowledge of $\rho,\sigma$, one has to try different $k$ until $\max_U\tr P_{\mathrm{rank}-k}(\rho-\sigma)$ reaches a maximum with respect to $k$. Then one have the following intuitive (but somehow naive) algorithm for trace distance estimation:
\vspace{0cm}
\begin{algorithm}[H] 
\caption{naive Variational Trace Distance Estimation (nVTDE)}
\begin{algorithmic} \label{alg:nvtde}
\REQUIRE $n$-qubit quantum states $\rho$ and $\sigma$, circuit ansatz of unitary $U_{AR}(\bm\theta)$, number of iterations ITR;
\ENSURE an estimation of $D(\rho,\sigma)$.

\FOR{$k=1,\dots,2^n-1$}
\STATE Initialize parameters $\bm\theta$.

\FOR{itr $=1,\ldots,$ ITR}

\STATE Prepare $k$ quantum states $\ket{j},~j\in[1,k]$;
 
\STATE Apply $U(\bm\theta)$ respectively to all $\ket{j}$, and obtain $\cL_4^k=\sum_{j=1}^{k}(\tr U\rho U^\dagger\proj{j}-\tr U\sigma U^\dagger\proj{j})$;
 
\STATE Maximize the loss function $\cL_4^k$ and update parameters of $\bm\theta$;


\ENDFOR
\IF{$\max_{\bm\theta}\cL_4^k(\bm\theta)$ does not increase with $k$}
\STATE break
\ENDIF

\ENDFOR


\STATE Output $\cL_4$ as the trace distance estimation;
\end{algorithmic}
\end{algorithm}

Algorithm \ref{alg:nvtde}, though requires no ancillary qubit, is inefficient for a large $k$. From the perspective of quantum resources, Algorithm \ref{alg:nvtde} needs at least $k$ computational basis quantum states. From the perspective of time, Algorithm \ref{alg:nvtde} also needs $k$ rounds of optimization. In the worst case, $k=2^n-1$.

We remark that, if one of the two unknown states is pure (suppose $\rho$ is the pure state), then $\rho-\sigma$ has at most one positive eigenvalue. Explicitly, we have the following lemma:

\begin{lemma}
If $\rho$ is a pure state and $\sigma$ is an arbitrary state, then $\rho-\sigma$ has at most one positive eigenvalue.
\end{lemma} 

\begin{proof} According to the Weyl's inequality in matrix theory (see \cite{weyl1912asymptotische}), for Hermitian matrices $A,B$, $C=A+B$ , we have 
\begin{equation}
    c^{j+k-1}\le a^j+b^k\le c^{j+k-d},
\end{equation}
where $a^j,b^j$ and $c^j$ are the $j-$th largest eigenvalue of $A,B$ and $C$, respectively, and $d$ is the dimension. Let $k=d=2^n,A=\rho-\sigma,B=\sigma$ and use the second inequality:
\begin{equation}
a^j+b^d\le c^{j}.
\end{equation}

We use the fact that $\sigma$ is positive (which implies $b^d\ge0$), and Rank$(\rho)=1$ (which implies $c^j=0$ for $j\ge2$). Finally we have
\begin{equation}
a^j\le0, \forall j\ge2,
\end{equation}
which complete the proof that $\rho-\sigma$ has at most one positive eigenvalue.
\end{proof}

\update{The above lemma can also be proved if we consider whether the support of $\rho$ is in the support of $\sigma$ and analyze each case separately.}

Thus, if one of the two states is pure, we only have to prepare one rank-1 projector $\proj{0}$.

\begin{figure}[h]
    \centering
    \includegraphics[width=0.5\textwidth]{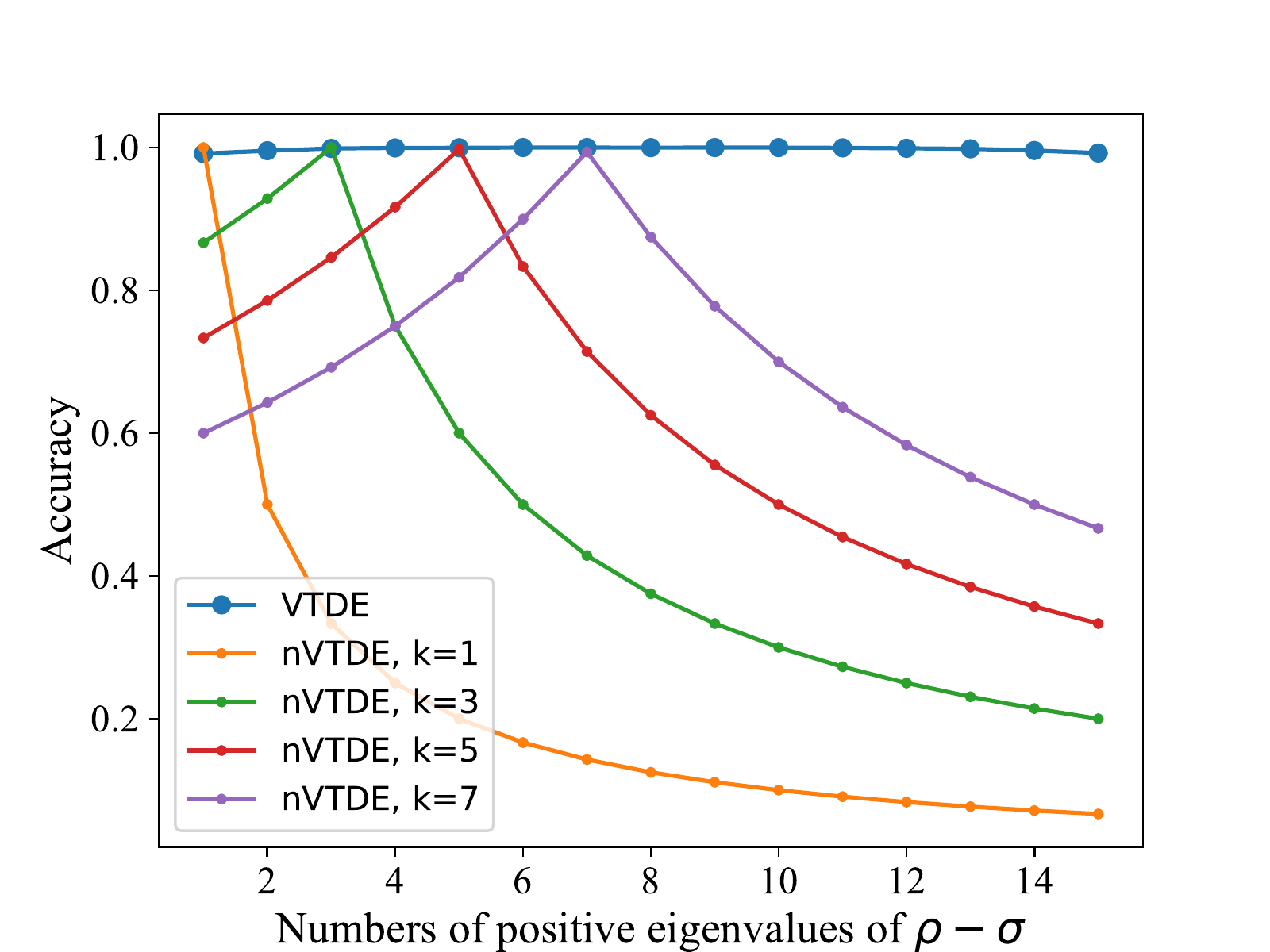}
    \caption{Trace distance learned by VTDE and nVTDE versus the number of positive eigenvalues of $H=\frac{1}{2}(\rho-\sigma)$. The curves show that the performance of VTDE is stable with different numbers of positive eigenvalues. In contrast, the performance of nVTDE relies heavily on the difference between $k$ and the number of positive eigenvalues.}
    \label{fig:naive}
\end{figure}

Fig. \ref{fig:naive} shows the behavior of VTDE and nVTDE for $\rho-\sigma$ with different number of positive eigenvalues. In the $n=4$ case, the number of positive eigenvalues can range from 1 to $2^4-1$. Fig. \ref{fig:naive} shows that the performance of VTDE is stable with different numbers of positive eigenvalues, while the performance of nVTDE relies heavily on the difference between $k$ and the number of positive eigenvalues, as analyzed above.

\section{Analytical Gradient for VFE}
\label{proof:gradient_VFE}
Recall the loss function for VFE,
\begin{align}
\cL(\bm \t) = \, _{AR}\bra{\psi}I_A\otimes U_R(\bm \t) \ket{\phi}_{AR}.
\end{align}
Consider the parametrized quantum circuit $U(\bm \theta)= \Pi_{i = n}^{1}U_i(\theta_i)$, where each gate $U_{i}$ is either fixed, e.g., C-NOT gate, or generated by $U_{i}=e^{-i\theta_i H_{i}/2}$. For convenience, denote $U_{i:j} = U_i \cdots U_j$. We can write the cost function as
\begin{align}
\cL(\bm \t) = \, _{AR}\bra{\psi}I_A\otimes U_{1:n}(\theta_{1:n}) \ket{\phi}_{AR}.
\end{align}
Absorb all gates except the one we want to calculate the gradient into $\langle\psi|$ and $|\phi\rangle$:
\begin{align}
\cL(\bm \t) &= \, _{AR}\bra{\psi}(I_A\otimes U_1)\cdots (I_A\otimes U_j) \cdots (I_A\otimes U_n) \ket{\phi}_{AR} \notag \\
& = \, _{AR}\bra{\psi'} (I_A\otimes U_j)  \ket{\phi'}_{AR},
\end{align}
where $\ket{\phi'}_{AR} = (I_A\otimes U_{j+1:n}) \ket{\phi}_{AR}$ and $\bra{\psi'}_{AR} = \bra{\psi}_{AR} (I_A\otimes U_{1:j-1})$.
With the following property,
\begin{align}
\frac{\partial U_{1:n}}{\partial \theta_j} = -\frac{i}{2}U_{1:j-1} (H_j U_j) U_{j+1:n},
\end{align}
the gradient is calculated to be
\begin{align}
\frac{\partial \cL(\bm \t)}{\partial \theta_j} = -\frac{i}{2} \, _{AR}\bra{\psi'} (I_A\otimes H_jU_j) \ket{\phi'}_{AR}.
\end{align}
We can absorb the Pauli product $H_j$ into  $U_j$ by an rotation on $\theta_j \rightarrow \theta_j \pm \pi$:
\begin{align}
U_j(\pm \pi) & = e^{\mp i\pi H_j/2} \notag\\
& = \mp i H_j,
\end{align}
which leads to
\begin{align}
\frac{\partial \cL(\bm \t)}{\partial \theta_j} & = \frac{1}{2} \, _{AR}\bra{\psi} (I_B\otimes U_j(\theta_j + \pi))  \ket{\phi}_{AR} \notag\\
& = \frac{1}{2} L(\theta_j + \pi).
\end{align}

\section{Quantum Circuit for VFE}
\label{cir:VFE}
There are two types of parametrized quantum circuit used in VFE. The first type $U_{R}(\bm \t)$ only acts on the ancillary system and aims at fulfilling Algorithm \ref{alg:vfe}. In this specific study, it is taken to be the $U_3$ rotation gate on Bloch sphere, 2-qubit universal gate~\cite{vidal2004universal}, and 3-qubit universal gate~\cite{vatan2004realization} for $n_A = \{1,2,3\}$ respectively. The second type serves for Algorithm \ref{alg:vqsl} purification learning and is shown in Fig.~\ref{fig:VQSL-Ansatz}.

\begin{figure}[h]
\[\Qcircuit @C=0.8em @R=0.5em{
&\gate{U_3(\t_1,\t_2,\t_3)}&\ctrl{1} &\qw       &\targ  &\qw  &\cdots\\ 
&\gate{U_3(\t_4,\t_5,\t_6)}&\targ    &\ctrl{1}  &\qw  &\qw  &\cdots\\
&\gate{U_3(\t_7,\t_8,\t_9)}&\qw      &\targ     &\ctrl{-2} &\qw^{\quad \quad \quad \,\times L}  &\cdots
\gategroup{1}{2}{3}{5}{2.07em}{--}
}\]
\caption{Parameterized ansatz $U_{AR}(\bm\theta)$ used for state learning. Depth $L$ denotes the number of repetitions of the same block (dashed-line box). Each block consists of a column of general $U_3(\t,\phi,\varphi) = R_z(\phi)R_y(\t)R_z(\varphi)$ rotations for each qubit, followed with a circular layer of CNOT gates. Notice that the number of parameters in this ansatz increases linearly with the circuit depth $L$ and the number of qubits.}
\label{fig:VQSL-Ansatz}
\end{figure}
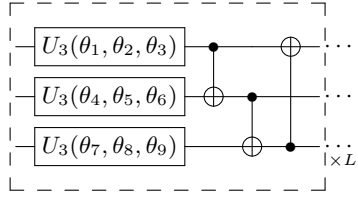

\end{document}